\documentclass[12pt,letterpaper]{article}%
\usepackage{latexsym,amsmath,amssymb}
\usepackage{epsfig,graphicx}
\usepackage{amsmath}
\usepackage{amsfonts}
\usepackage{amssymb}
\usepackage{graphicx}
\usepackage{color}
\usepackage[authoryear,round]{natbib}
\usepackage{hyperref}
\usepackage{geometry}
\usepackage[doublespacing]{setspace}
\providecommand{\U}[1]{\protect\rule{.1in}{.1in}}

\newtheorem{lemma}{Lemma}

\newtheorem{proposition}{Proposition}
\newtheorem{remark}{Remark}

\newenvironment{proof}[1][Proof]{\textbf{#1.} }{\ \rule{0.5em}{0.5em}}
\hypersetup{colorlinks=true,citecolor=blue}

\begin{document}
\title{Stratifying on Treatment Status}
\author{Jinyong Hahn\thanks{Department of Economics, UCLA, Los Angeles, CA 90095, USA. Email: hahn@econ.ucla.edu.}\\UCLA
\and John Ham\thanks{Department of Economics, NYU Abu Dhabi, Saadiyat Campus, P.O. Box 129188, Abu Dhabi, United Arab Emirates. Email: jch18@nyu.edu.}\\NYU Abu Dhabi
\and Geert Ridder\thanks{Department of Economics, USC, Los Angeles, CA 90089, USA. Email: ridder@usc.edu.}\\USC
\and Shuyang Sheng\thanks {Shenzhen Finance Institute, School of Management and Economics, The Chinese University of Hong Kong, Shenzhen, 2001 Longxiang Boulevard, Longgang District, Shenzhen, China. Email: shengshuyang@cuhk.edu.cn.}\\CUHK-Shenzhen}
\date{\today}
\maketitle

\begin{abstract} 
We study the estimation of treatment effects using samples stratified by treatment status. Standard estimators of the average treatment effect and the local average treatment effect are inconsistent in this setting. We propose consistent estimators and characterize their asymptotic distributions.  

\end{abstract}

\newpage

\section{Introduction}

Sample surveys are usually not simple random samples. Often the population is
partitioned into strata, where the subpopulation in each stratum is sampled
with equal probability, but the inclusion probability is unequal between
strata. Stratified sampling is preferred for several reasons. First, it can
result in more precise estimates 
\citep{Domencich_McFadden_1975}. Second, it may over-represent groups of
particular interest. Third, it may arise naturally from the combination of
independent data sets into a single sample \citep{Ridder_Moffitt_2007}.

Many studies select strata based on treatment status. For example,
\citet{Azoulay_2010} analyze the impact of collaborator quality on research
productivity. They compare researchers whose superstar collaborators died with
those whose superstar collaborators did not. Their dataset includes all researchers with a deceased superstar co-author and
a random sample of researchers with a living one. Similarly, \citet{Ham_Khan_2023} assess the effectiveness of a new teaching
approach, JAAGO, against government and non-governmental schools in urban
Bangladesh. The new approach is implemented in only two schools in Dhaka, so
they use the population of JAAGO students alongside a random sample of
students from other schools.

The classical literature showed in parametric settings that stratification on
endogenous variables leads to biased estimates,\footnote{Stratification on
endogenous variables was first discussed in discrete choice models under the
name choice-based sampling \citep{Manski_Lerman_1977,Manski_McFadden_1981,Hausman_Wise_1981}.} whereas
stratification on exogenous variables typically does not result in a bias. We
show that stratification on unconfounded treatment results in a biased estimate
of the average treatment effect (ATE). This may seem to contradict the results in
\cite{Heckman_Todd_2009}, who argued that \textquotedblleft matching and
selection procedures can still be applied when the propensity score is
estimated on unweighted choice based samples\textquotedblright.

The seeming contradiction is resolved by observing that although the conditional average treatment effect (CATE) given confounders $X$ can be identified (from the conditional distributions
of outcome $Y$ given $X$ and treatment status), the \textit{average} of the CATE over the distribution of $X$ (or that of the propensity
score) in a stratified sample results in inconsistent estimates of the ATE. There is a simple fix to the problem, and a reweighted method can identify the ATE, provided that the population
fraction treated is known or can be estimated. In contrast, the average
treatment effect on the treated (ATT) is identified by conventional methods
despite stratification, even if the population fraction treated is unknown.

In addition, we consider the case of an endogenous treatment with a valid
binary instrument. We show that the standard Wald ratio estimator is not
consistent for the local average treatment effect (LATE) when the sample is
stratified on the treatment status. We propose a reweighted method to recover
the LATE, which requires knowledge of the fraction treated in the population.

We contribute to the literature by developing \textit{constructive}
identification results that can form a basis for consistent 
estimation. In particular, we present an analysis of the asymptotic distribution of various estimators under
stratified sampling, which is new to the
literature.\footnote{\citet{Correa_Tian_Bareinboim_2019} and \citet{Nabi2020}
provide generic identification results on the joint distribution of outcomes
and covariates for the treated and controls, but they do not have constructive
identification methods for estimating parameters of interest such as the ATE
or ATT. \citet{Hu2018} and \citet{Zhang2019} propose a reweighted method that is
the same as ours, but only for the ATE using covariate conditioning. We extend
this method to include other parameters of interest, such as the ATT and
LATE. \citet{Song2022} discuss the identification and estimation of ATE/ATT under unconfoundedness in the scenario where sampling is stratified based on the treatment and other discrete covariates. They focus on the case where the treatment probability in the population is known up to an interval. In contrast, we establish asymptotic results under the assumption that this probability is known. We also provide results for LATE, which is not covered in their paper.} 

\section{Unconfounded Treatment\label{sect.exog}}

We start with the case where the treatment $D$ is independent of potential outcomes $(Y_{0},Y_{1})$ given $X$.
Let $Y\equiv DY_{1}+(1-D)Y_{0}$ denote the observed outcome. Denote the
fraction $D=1$ in the sample by $\pi^{\ast}$ and in the population by $\pi$. In a stratified sample, we have $\pi^{\ast}\neq\pi$. 

For $d=0,1$, let $h(y|x,D=d)$ and $h^{\ast}(y|x,D=d)$ denote the conditional
densities of $Y$ given $X=x$ and $D=d$ in the population and sample,
respectively. Similarly, let $g(x|D=d)$ and $g^{\ast}(x|D=d)$ denote the
conditional densities of $X$ given $D=d$ in the population and sample. Because we sample randomly in the strata, the conditional distribution of $(Y,X)$ given $D$ is identical in the population and sample. Therefore for
$d=0,1$,
\begin{align}
h(y|x,D=d)  &  =h^{\ast}(y|x,D=d),\label{Ydensity}\\
g(x|D=d)  &  =g^{\ast}(x|D=d). \label{Xdensity}%
\end{align}

Sampling objects that do not condition on the treatment status, however, may
differ from the population counterparts. To understand their relationship, let
$\pi(x)\equiv\Pr(D=1|X=x)$ denote the propensity score in the population, and
$\pi^{\ast}(x)\equiv\Pr^{\ast}(D=1|X=x)$ the propensity score in the
stratified sample. Moreover, let $g(x)$ and $g^{\ast}(x)$ denote the densities
of $X$ in the population and stratified sample, respectively. By Bayes'
theorem
\begin{equation}
g^{\ast}(x|D=1)=g^{\ast}(x)\frac{\pi^{\ast}(x)}{\pi^{\ast}},\qquad g^{\ast
}(x|D=0)=g^{\ast}(x)\frac{1-\pi^{\ast}(x)}{1-\pi^{\ast}}, \label{g*Bayes}%
\end{equation}
and similar equations hold for the population counterparts. Therefore, we can
derive
\begin{equation}
\pi(x)=\frac{g(x|D=1)\pi}{g(x|D=1)\pi+g(x|D=0)(1-\pi)}=\frac{g^{\ast
}(x|D=1)\pi}{g^{\ast}(x|D=1)\pi+g^{\ast}(x|D=0)(1-\pi)}, \label{pscore}%
\end{equation}
from which we obtain
\begin{equation}
\pi(x)=\frac{\pi^{\ast}(x)\frac{\pi}{\pi^{\ast}}}{\pi^{\ast}(x)\frac{\pi}%
{\pi^{\ast}}+(1-\pi^{\ast}(x))\frac{1-\pi}{1-\pi^{\ast}}}. \label{pi}%
\end{equation}
Because $\pi^{\ast}(x)$ and $\pi^{\ast}$ are identified from the sampling
distribution of $(Y,D,X)$, the population propensity score $\pi(x)$ is
identified \textit{if and only }$\pi$ is known.\ Switching the population and
sample objects in (\ref{pscore}), we can derive similarly
\begin{equation}
\pi^{\ast}(x)=\frac{\pi(x)\frac{\pi^{\ast}}{\pi}}{\pi(x)\frac{\pi^{\ast}}{\pi
}+(1-\pi(x))\frac{1-\pi^{\ast}}{1-\pi}}. \label{pi*}%
\end{equation}
In addition, the population density $g(x)$ of $X$ satisfies%
\[
g(x)=g(x|D=1)\pi+g(x|D=0)(1-\pi)=g^{\ast}(x|D=1)\pi+g^{\ast}(x|D=0)(1-\pi),
\]
so by (\ref{g*Bayes})
\begin{equation}
g(x)=g^{\ast}(x)\left(  \pi^{\ast}(x)\frac{\pi}{\pi^{\ast}}+(1-\pi^{\ast
}(x))\frac{1-\pi}{1-\pi^{\ast}}\right)  . \label{g*tog}%
\end{equation}
Because $\pi^{\ast}(x)$, $g^{\ast}(x)$, and $\pi^{\ast}$ are identified from
the distribution of $(Y,D,X)$ in the stratified sample, the population density
$g(x)$ of $X$ can be identified \textit{if and only if} $\pi$ is known.

Equation (\ref{Ydensity}) implies that the conditional average treatment
effect (CATE) is identified by 
\begin{equation}
\beta(X)\equiv E[Y|D=1,X]-E[Y|D=0,X]=E^{\ast}[Y|D=1,X]-E^{\ast}[Y|D=0,X]\equiv
\beta^{\ast}(X), \label{matching-works}%
\end{equation}
where $E$ and $E^{\ast}$ denote the expectations taken with respect to
$h(y|X,D=d)$ and $h^{\ast}(y|X,D=d)$, which are equal. This is the sense in which conditioning
on $X$ \textquotedblleft works\textquotedblright\ \citep{Heckman_Todd_2009}.

\subsection{Average Treatment Effect (ATE)}

We now explore whether the success of conditioning translates into success in
identifying the ATE, $\beta\equiv E[Y_{1}-Y_{0}]$. By iterated expectations,
the ATE and its sample counterpart are
\begin{align*}
E[Y_{1}-Y_{0}]  &  =E[E[Y|D=1,X]-E[Y|D=0,X]],\\
E^{\ast}[Y_{1}-Y_{0}]  &  =E^{\ast}[E^{\ast}[Y|D=1,X]-E^{\ast}[Y|D=0,X]].
\end{align*}
Because $g^{\ast}(x)\neq g(x)$ by (\ref{g*tog}), the sample counterpart does
not identify the ATE even though $\beta(x)=\beta^*(x)$
\[
E[Y_{1}-Y_{0}]=\int\beta(x)g(x)dx\neq\int\beta^{\ast}(x)g^{\ast}(x)dx=E^{\ast
}[Y_{1}-Y_{0}].
\]
This result is intuitive and unsurprising. The problem is not the failure of
conditioning, but is due to
averaging the CATE over the wrong distribution of $X$.

The ATE cannot be identified by conditioning on the propensity score either.
Because the independence of the potential outcomes and $D$ given $X$ implies
independence given the propensity score, it follows by iterated expectations
that $E[Y_{1}-Y_{0}]=E[E[Y|D=1,\pi(X)]-E[Y|D=0,\pi(X)]]$. From (\ref{pi}) and
(\ref{pi*}), there is a 1-1 relation between $\pi(x)$ and $\pi^{\ast}(x)$. It
follows that the $\sigma$-algebras generated by these two random variables are
identical and thus conditioning on $\pi(X)$ and conditioning on $\pi^{\ast
}(X)$ are equivalent \citep{Heckman_Todd_2009}. Hence,
\begin{align*}
\delta(\pi(X))  &  \equiv E[Y|D=1,\pi(X)]-E[Y|D=0,\pi(X)]\\
&  =E^{\ast}[Y|D=1,\pi^{\ast}(X)]-E^{\ast}[Y|D=0,\pi^{\ast}(X)]\equiv
\delta^{\ast}(\pi^{\ast}(X)).
\end{align*}
Unfortunately, the averaging problem remains. We can see that
\[
E[Y_{1}-Y_{0}]=\int\delta(\pi(x))g(x)dx\neq\int\delta^{\ast}(\pi^{\ast
}(x))g^{\ast}(x)dx=E^{\ast}[Y_{1}-Y_{0}]
\]
because $g^{\ast}(x)\neq g(x)$ by (\ref{g*tog}). Conditioning only works
partially because it does not solve the problem of averaging. Therefore, we should be
careful interpreting \citet[p.S231]{Heckman_Todd_2009}'s observation that \textquotedblleft matching and selection procedures can
identify population treatment effects using misspecified estimates of
propensity scores fit on choice-based samples\textquotedblright\ when
estimating the ATE.

The problem remains if we use the inverse propensity score weighting (IPW).
IPW identifies the ATE by
\[
E[Y_{1}-Y_{0}]=E\left[  \frac{DY}{\pi(X)}-\frac{(1-D)Y}{1-\pi(X)}\right]  ,
\]
which is not identified by the sample counterpart
\[
E\left[  \frac{DY}{\pi(X)}-\frac{(1-D)Y}{1-\pi(X)}\right]  =E[\beta(X)]\neq
E^{\ast}[\beta^{\ast}(X)]=E^{\ast}\left[  \frac{DY}{\pi^{\ast}(X)}%
-\frac{(1-D)Y}{1-\pi^{\ast}(X)}\right]
\]
again due to improper averaging, even though $\beta(x)=\beta^{\ast}(x)$. 

If we know $\pi$, we can identify the ATE using reweighting based on
(\ref{g*tog})
\begin{align}
E[Y_{1}-Y_{0}]  &  =\int\beta^{\ast}(x)g(x)dx\nonumber\\
&  =\int\beta^{\ast}(x)\left(  \pi^{\ast}(x)\frac{\pi}{\pi^{\ast}}%
+(1-\pi^{\ast}(x))\frac{1-\pi}{1-\pi^{\ast}}\right)  g^{\ast}(x)dx\nonumber\\
&  =E^{\ast}\left[  \beta^{\ast}(X)\left(  D\frac{\pi}{\pi^{\ast}}%
+(1-D)\frac{1-\pi}{1-\pi^{\ast}}\right)  \right]  . \label{ATE}%
\end{align}
If we condition on the propensity score, we can also identify the ATE through
reweighting
\begin{align}
E[Y_{1}-Y_{0}]  &  =\int\delta^{\ast}(\pi^{\ast}(x))g(x)dx\nonumber\\
&  =\int\delta^{\ast}(\pi^{\ast}(x))\left(  \pi^{\ast}(x)\frac{\pi}{\pi^{\ast
}}+(1-\pi^{\ast}(x))\frac{1-\pi}{1-\pi^{\ast}}\right)  g^{\ast}%
(x)dx\nonumber\\
&  =E^{\ast}\left[  \delta^{\ast}(\pi^{\ast}(X))\left(  D\frac{\pi}{\pi^{\ast
}}+(1-D)\frac{1-\pi}{1-\pi^{\ast}}\right)  \right]  . \label{ATE-PS}%
\end{align}
In addition, the ATE can be identified by a reweighted version of IPW%
\begin{equation}
E^{\ast}\left[  \left(  \frac{DY}{\pi^{\ast}(X)}-\frac{(1-D)Y}{1-\pi^{\ast
}(X)}\right)  \left(  \pi^{\ast}(X)\frac{\pi}{\pi^{\ast}}+(1-\pi^{\ast
}(X))\frac{1-\pi}{1-\pi^{\ast}}\right)  \right]  . \label{IPW}%
\end{equation}

In sum, conventional methods do not identify the ATE, but the problem can be
overcome by reweighting the observations according to the population
distribution of $X$.

\subsection{Average Treatment Effect on the Treated (ATT)}

Next we examine the ATT, $\gamma\equiv E[Y_{1}-Y_{0}|D=1]$. Note that by
(\ref{g*Bayes}) and (\ref{matching-works})
\begin{align}
E[Y_{1}-Y_{0}|D=1]  &  =\int\beta(x)g(x|D=1)dx\nonumber\\
&  =\int\beta^{\ast}(x)g^{\ast}(x|D=1)dx=E^{\ast}[Y_{1}-Y_{0}|D=1].
\label{ATT}%
\end{align}
The ATT is identified despite stratification. We can also identify the ATT
conditioning on the propensity score because
\begin{align}
E[Y_{1}-Y_{0}|D=1]  &  =\int\delta(\pi(x))g(x|D=1)dx\nonumber\\
&  =\int\delta^{\ast}(\pi^{\ast}(x))g^{\ast}(x|D=1)dx=E^{\ast}[Y_{1}%
-Y_{0}|D=1]. \label{ATT-PS}%
\end{align}
Lastly, note that by iterated expectations
\[
E[Y_{1}-Y_{0}|D=1]=\frac{1}{\pi}E\left[  \left(  \frac{DY}{\pi(X)}%
+\frac{(1-D)Y}{1-\pi(X)}\right)  \pi(X)\right]  ,
\]
whose sample counterpart identifies the ATT
\[
\frac{1}{\pi^{\ast}}E^{\ast}\left[  \left(  \frac{DY}{\pi^{\ast}(X)}%
+\frac{(1-D)Y}{1-\pi^{\ast}(X)}\right)  \pi^{\ast}(X)\right]  =E^{\ast}%
[\beta^{\ast}(X)|D=1]=E[\beta(X)|D=1]=E[Y_{1}-Y_{0}|D=1].
\]
Hence, the ATT can be identified using IPW as well.

In sum, the ATT is identified by conventional methods. Stratification has no
impact on the averaging because the distribution of $X$ given $D=1$ is the
same in the population and stratified sample.

\subsection{Asymptotic Distribution}

\subsubsection{ATE Estimators}

Assuming that the population fraction $\pi$ is known, equation (\ref{ATE}) suggests that we can estimate the ATE $\beta$ by a
semiparametric estimator $\hat{\beta}$ based on the moments
\begin{align}
E^{\ast}[Y-\beta_{1}^{\ast}(X)|D=1,X]  &  =0\nonumber\\
E^{\ast}[Y-\beta_{0}^{\ast}(X)|D=0,X]  &  =0\nonumber\\
E^{\ast}\left[  \left(  D\frac{\pi}{\pi^{\ast}}+(1-D)\frac{1-\pi}{1-\pi^{\ast
}}\right)  (\beta_{1}^{\ast}(X)-\beta_{0}^{\ast}(X))-\beta\right]   &  =0,
\label{ATE-estimator}%
\end{align}
where $\beta_{d}^{\ast}(X)\equiv E^{\ast}[Y|D=d,X]$, $d=0,1$. Alternatively,
equation (\ref{ATE-PS}) suggests a semiparametric estimator using the
propensity score
\begin{align}
E^{\ast}[D-\pi^{\ast}(X)|X]  &  =0\nonumber\\
E^{\ast}[Y-\delta_{1}^{\ast}(\pi^{\ast}(X))|D=1,\pi^{\ast}(X)]  &
=0\nonumber\\
E^{\ast}[Y-\delta_{0}^{\ast}(\pi^{\ast}(X))|D=0,\pi^{\ast}(X)]  &
=0\nonumber\\
E^{\ast}\left[  \left(  D\frac{\pi}{\pi^{\ast}}+(1-D)\frac{1-\pi}{1-\pi^{\ast
}}\right)  (\delta_{1}^{\ast}(\pi^{\ast}(X))-\delta_{0}^{\ast}(\pi^{\ast
}(X)))-\beta\right]   &  =0, \label{ATE-PS-estimator}%
\end{align}
where $\delta_{d}^{\ast}(\pi^{\ast}(X))\equiv E^{\ast}[Y|D=d,\pi^{\ast}(X)]$,
$d=0,1$.\footnote{The moment conditions (\ref{ATE-estimator}) and
(\ref{ATE-PS-estimator}) can be understood to be semiparametric
generalization of the parametric model considered by \citet{Nevo2003}, e.g.}
Applying \citet{Newey_1994}, we derive the influence functions of these ATE
estimators, presented in Propositions \ref{prop:ATE} and \ref{prop:ATE-PS}.\footnote{All the proofs are provided in the Supplementary Appendix, which is
available upon request.}  Propositions \ref{prop:ATE} and \ref{prop:ATE-PS} are both predicated on the assumption that the population fraction $\pi$ is known. The sample fraction $\pi^{\ast}$ may or may not be known, and the two propositions consider both cases.\footnote{The representations in (\ref{ATE-estimator}) and
(\ref{ATE-PS-estimator}) assume that $\pi^{\ast}$ is known. If $\pi^{\ast}$ is
unknown, we can add one more moment $E^{\ast}[D-\pi^{\ast}]=0$ to each system.} 

\begin{proposition}
\label{prop:ATE}If $\pi^{\ast}$ is known, the influence function of the ATE
estimator based on (\ref{ATE-estimator}) is
\begin{align}
&  \left(  D\frac{\pi}{\pi^{\ast}}+(1-D) \frac{1-\pi}{1-\pi^{\ast}}\right)
(\beta_{1}(X)-\beta_{0}(X)) -\beta\nonumber\\
&  +\left(  \pi^{\ast}(X) \frac{\pi}{\pi^{\ast}}+(1-\pi^{\ast}(X)) \frac
{1-\pi}{1-\pi^{\ast}}\right)  \left(  \frac{D}{\pi^{\ast}(X) }(Y-\beta_{1}(X))
-\frac{1-D}{1-\pi^{\ast}( X)}(Y-\beta_{0}(X)) \right)  , \label{ATE-infl}%
\end{align}
where $\beta_{d}(X)\equiv E[Y|D=d,X]$, $d=0,1$. If $\pi^{\ast}$ is estimated, the influence function is the sum of (\ref{ATE-infl}) and
\begin{equation}
E^{\ast}\left[  \left(  -\pi^{\ast}(X) \frac{\pi}{(\pi^{\ast})^{2}}+(
1-\pi^{\ast}(X)) \frac{1-\pi}{(1-\pi^{\ast})^{2}}\right)  (\beta_{1}
(X)-\beta_{0}(X)) \right]  (D-\pi^{\ast}). \label{ATE-infl-adjust}%
\end{equation}

\end{proposition}

\begin{proposition}
\label{prop:ATE-PS}The influence function of the ATE estimator based on
(\ref{ATE-PS-estimator}) equals (\ref{ATE-infl}) if $\pi^{\ast}$ is known, and
equals the sum of (\ref{ATE-infl}) and (\ref{ATE-infl-adjust}) if $\pi^{\ast}$
is estimated.
\end{proposition}

The results mirror those in the non-stratified case, where various ATE estimators have the same asymptotic distribution. We can derive the asymptotic
variance of the ATE estimators using (\ref{ATE-infl}) if $\pi^{\ast}$ is
known. Denote $\epsilon_{d}\equiv Y-\beta_{d}(X)$ and $\sigma_{d}^{2}(X)
\equiv E^{\ast}[\epsilon_{d}^{2}|X]$, $d=0,1$, and $w(X)\equiv\pi^{\ast}(X)
\frac{\pi}{\pi^{\ast}}+(1-\pi^{\ast}(X)) \frac{1-\pi}{1-\pi^{\ast}}$. The
asymptotic variance of $\sqrt{n}(\hat{\beta}-\beta)$ is
\begin{equation}
E^{\ast}\left[  \left(  \left(  D\frac{\pi}{\pi^{\ast}}+(1-D) \frac{1-\pi
}{1-\pi^{\ast}}\right)  \beta(X) -\beta\right)  ^{2} +\frac{w(X)^{2}\sigma
_{1}^{2}(X)}{\pi^{\ast}(X)} +\frac{w(X)^{2}\sigma_{0}^{2}(X)}{1-\pi^{\ast}
(X)}\right]. \label{ATE-var}
\end{equation}

\subsubsection{ATT Estimators}

Equations (\ref{ATT}) and (\ref{ATT-PS}) serve as the basis for estimating the
ATT. Because $\beta^{\ast}(X)=E^{\ast}[Y-\beta_{0}^{\ast}(X)|D=1,X]$, where
$\beta_{0}^{\ast}(X)\equiv E^{\ast}[Y|D=0,X]$, it follows by iterated
expectations that $E^{\ast}[\beta^{\ast}(X)|D=1]=E^{\ast}[Y-\beta_{0}^{\ast
}(X)|D=1]=E^{\ast}[D(Y-\beta_{0}^{\ast}(X))]/\pi^{\ast}$. This suggests an
estimator of $\gamma$ of the form
\begin{equation}
\hat{\gamma}\equiv\frac{1}{n^{-1}\sum_{i=1}^{n}D_{i}}\cdot\frac{1}{n}%
\sum_{i=1}^{n}D_{i}(Y_{i}-\hat{\beta}_{0}(X_{i})), \label{ATT-estimator}%
\end{equation}
where $\hat{\beta}_{0}(\cdot)$ is a nonparametric estimator of $\beta
_{0}^{\ast}(\cdot)$. This estimator subtracts the predicted outcome for the
counterfactual from the outcome of a treated unit.
Although intuitive, we are not aware of any references that establish the
asymptotic distribution of such an estimator.

Define $\delta_{0}^{\ast}(\pi^{\ast}(X))\equiv E^{\ast}[Y|D=0,\pi^{\ast}(X)]$.
A similar estimator can be constructed based on the propensity score
\begin{equation}
\hat{\gamma}_{PS}\equiv\frac{1}{n^{-1}\sum_{i=1}^{n}D_{i}}\cdot\frac{1}{n}%
\sum_{i=1}^{n}D_{i}(Y_{i}-\hat{\delta}_{0}(\hat{\pi}(X_{i}))),
\label{ATT-PS-estimator}
\end{equation}
where $\hat{\delta}_{0}(\cdot)$ and $\hat{\pi}(\cdot)$ are nonparametric
estimators of $\delta_{0}^{\ast}(\cdot)$ and $\pi^{\ast}(\cdot)$,
respectively. Propositions \ref{prop:ATT} and \ref{prop:ATT-PS} derive the
asymptotic variances of these ATT estimators.

\begin{proposition}
\label{prop:ATT}If $\pi^{\ast}$ is known, the asymptotic variance of $\sqrt
{n}(\hat{\gamma}-\gamma)$ is equal to
\begin{equation}
E^{\ast}\left[  \frac{\pi^{\ast}(X)(\beta(X)-\gamma)^{2}}{(\pi^{\ast})^{2}%
}+\frac{\pi^{\ast}(X) \sigma_{1}^{2}(X)}{(\pi^{\ast})^{2}}+\frac{\pi^{\ast
}(X)^{2}\sigma_{0}^{2}(X) }{(\pi^{\ast})^{2}(1-\pi^{\ast}(X))}\right]  .
\label{ATT-var}%
\end{equation}

\end{proposition}

\begin{proposition}
\label{prop:ATT-PS}If $\pi^{\ast}$ is known, the asymptotic variance of
$\sqrt{n}(\hat{\gamma}_{PS}-\gamma)$ is the same as in (\ref{ATT-var}).
\end{proposition}

The ATT estimators in (\ref{ATT-estimator}) and (\ref{ATT-PS-estimator}) have
the same asymptotic variance in (\ref{ATT-var}), which is equal to the
asymptotic variance bound based on the stratified sample.\footnote{Note that
(\ref{ATT-var}) is different from the asymptotic variance based on an
unstratified population. While the ATT can be consistently estimated, the
asymptotic variance is impacted by stratification.} See \cite{Hahn_1998}. So
both ATT estimators are efficient.\footnote{The fact that the ATE (or ATT)
estimators that condition on either the propensity score or the covariates
have the same asymptotic variance indicates that there is no efficiency gain
from using the propensity score in stratified samples. \citet{HR_2013} come to
the same conclusion in random sampling.}

\section{Endogenous Treatment}

In this section, we consider stratification on an endogenous treatment. We
have a binary instrument $Z$ and assume no covariates to focus on essentials.
\citet{AIR_1996} show that the local average treatment effect (LATE) is
identified by the Wald ratio $(E[Y|Z=1]-E[Y|Z=0])/(E[D|Z=1]-E[D|Z=0])$. If the
sample is stratified on the treatment status, we have
\begin{align*}
E[YZ|D=d]=  &  E^{\ast}[YZ|D=d]\\
E[Y(1-Z)|D=d]=  &  E^{\ast}[Y(1-Z)|D=d]\\
E[Z|D=d]=  &  E^{\ast}[Z|D=d],
\end{align*}
for $d=0,1$. However, because $\pi\neq\pi^{\ast}$, unconditional expectations
in the sample do not match those in the population. We can see that
\begin{align*}
E[Y|Z=1]-E[Y|Z=0]=  &  \frac{E[YZ]}{E[Z]}-\frac{E[Y(1-Z)]}{1-E[Z]}\\
\neq &  \frac{E^{\ast}[YZ]}{E^{\ast}[Z]}-\frac{E^{\ast}[Y(1-Z)]}{1-E^{\ast
}[Z]}=E^{\ast}[Y|Z=1]-E^{\ast}[Y|Z=0],
\end{align*}
and similarly $E[D|Z=1]-E[D|Z=0]\neq E^{\ast}[D|Z=1]-E^{\ast}[D|Z=0]$.
Therefore, the sample Wald ratio does not identify the population Wald ratio.

Observe that the population Wald ratio is the solution for $\beta$ in the
moment condition
\[
E\left\{
\begin{bmatrix}
1\\
Z
\end{bmatrix}
(Y-\alpha-\beta D)\right\}  =0.
\]
Proposition \ref{prop:weighting} suggests that we can identify the population
Wald ratio using reweighting.

\begin{proposition}
\label{prop:weighting}For any function $\varphi(Y,Z,D)$, we have
\[
E[\varphi(Y,Z,D)]=E^{\ast}\left[  \left(  D\frac{\pi}{\pi^{\ast}}%
+(1-D)\frac{1-\pi}{1-\pi^{\ast}}\right)  \varphi(Y,Z,D)\right]  .
\]

\end{proposition}

Proposition \ref{prop:weighting} implies that the population Wald ratio is the
solution for $\beta$ in the reweighted sample moment condition
\[
E^{\ast}\left\{  \left(  D\frac{\pi}{\pi^{\ast}}+(1-D)\frac{1-\pi}{1-\pi
^{\ast}}\right)
\begin{bmatrix}
1\\
Z
\end{bmatrix}
(Y-\alpha-\beta D)\right\}  =0.
\]
This fix requires knowledge of $\pi$.\footnote{\citet{Froelich_2007} considers
a nonparametric estimator of LATE that allows for covariates.}

\begin{remark}
In the unconfounded case without $X$, i.e., under random assignment,
stratification on $D$ does not pose any issue, because without $X$ there is no
need for averaging. In the endogenous case, however, the bias is present even
without $X$. The intuition can be found in the literature on stratification on
endogenous variables (e.g, \citealp{Hausman_Wise_1981}).
\end{remark}

\section{Conclusion}

For unconfounded treatment, stratification on the treatment status does not
bias conditional treatment effects. It does bias the ATE, which can be
corrected through reweighting if we know the fraction treated in the
population. Stratification does not bias the ATT. In the case of endogenous
treatment with an instrument, stratification biases the LATE. We propose a
reweighted method to identify the LATE.

\bibliography{StraTreatment}

\newpage
\appendix
\setcounter{page}{1} \setcounter{section}{1} \renewcommand{\thesection}{\Alph{section}}

\begin{center}
	{ {\Large Online Supplementary Material for \textquotedblleft Stratifying on Treatment
			Status\textquotedblright}}
\end{center}

\subsection{Useful Lemmas}

The following lemmas are used to prove the propositions in the paper.

\begin{lemma}
	\label{lem:ATE-adjust}The adjustment for the estimation of $(\beta
	_{1}(X),\beta_{0}(X))$ in (\ref{ATE-estimator}) is
	\begin{equation}
		\left(  \pi^{\ast}(X)\frac{\pi}{\pi^{\ast}}+(1-\pi^{\ast}(X))\frac{1-\pi
		}{1-\pi^{\ast}}\right)  \left(  \frac{D}{\pi^{\ast}(X)}\left(  Y-\beta
		_{1}(X)\right)  -\frac{1-D}{1-\pi^{\ast}(X)}(Y-\beta_{0}(X))\right)  .
		\label{ATE-adjust}%
	\end{equation}
	
\end{lemma}

\begin{lemma}
	\label{lem:ATE-PS-adjust}The adjustment for the estimation of $\pi^{\ast}(X)$
	in (\ref{ATE-PS-estimator}) is
	\begin{align}
		&  \left(  D\frac{\pi}{\pi^{\ast}}+(1-D)\frac{1-\pi}{1-\pi^{\ast}}\right)
		(\beta(X)-\delta^{\ast}(\pi^{\ast}(X)))\nonumber\\
		&  -\frac{\pi^{\ast}(X)\frac{\pi}{\pi^{\ast}}+(1-\pi^{\ast}(X))\frac{1-\pi
			}{1-\pi^{\ast}}}{\pi^{\ast}(X)}D(\beta_{1}(X)-\delta_{1}^{\ast}(\pi^{\ast
		}(X)))\nonumber\\
		&  +\frac{\pi^{\ast}(X)\frac{\pi}{\pi^{\ast}}+(1-\pi^{\ast}(X))\frac{1-\pi
			}{1-\pi^{\ast}}}{1-\pi^{\ast}(X)}(1-D)(\beta_{0}(X)-\delta_{0}^{\ast}%
		(\pi^{\ast}(X))). \label{ATE-PS-pi}%
	\end{align}
	
\end{lemma}

\begin{lemma}
	\label{lem:ATT-adjust}The adjustment for the estimation of $\beta_{0}(X)$ in
	$\frac{1}{n}\sum_{i=1}^{n}D_{i}(Y_{i}-\hat{\beta}_{0}(X_{i}))$ is
	\[
	-\frac{\pi^{\ast}(X)}{1-\pi^{\ast}(X)}(1-D)(Y-\beta_{0}(X)).
	\]
	
\end{lemma}

\begin{lemma}
	\label{lem:ATT-PS-adjust}The adjustment for the estimation of $\pi^{\ast}(X)$
	in $\frac{1}{n}\sum_{i=1}^{n}D_{i}(Y_{i}-\hat{\delta}_{0}(\hat{\pi}(X_{i})))$
	is
	\[
	-\frac{1}{1-\pi^{\ast}(X)}(\beta_{0}(X)-\delta_{0}^{\ast}(\pi^{\ast
	}(X)))(D-\pi^{\ast}(X)).
	\]
	
\end{lemma}

\subsection{Proofs of the Lemmas}

\begin{proof}
	[Proof of Lemma \ref{lem:ATE-adjust}]By (\ref{ATE-estimator}) we can write the
	ATE as a linear functional in the conditional means of $Y$ given $X$ for the
	treated and controls
	\begin{align*}
		\beta &  =\mathbb{E}^{\ast}\left[  \left(  D\frac{\pi}{\pi^{\ast}}%
		+(1-D)\frac{1-\pi}{1-\pi^{\ast}}\right)  (h_{1}(X)-h_{2}(X))\right] \\
		&  =\mathbb{E}^{\ast}\left[  \left(  \pi^{\ast}(X)\frac{\pi}{\pi^{\ast}%
		}+(1-\pi^{\ast}(X))\frac{1-\pi}{1-\pi^{\ast}}\right)  (h_{1}(X)-h_{2}%
		(X))\right] \\
		&  =\mathbb{E}^{\ast}[w(X)(h_{1}(X)-h_{2}(X))].
	\end{align*}
	In this expression $h_{1}(X)\equiv\mathbb{E}^{\ast}[Y|D=1,X]=\beta_{1}(X)$,
	$h_{2}(X)\equiv\mathbb{E}^{\ast}[Y|D=0,X]=\beta_{0}(X)$, and $w(X)\equiv
	\pi^{\ast}(X)\frac{\pi}{\pi^{\ast}}+(1-\pi^{\ast}(X))\frac{1-\pi}{1-\pi^{\ast
	}}$. Therefore the ATE satisfies the moment equation $\mathbb{E}^{\ast
	}[m(X,\beta,h_{1},h_{2})]=0$, where
	\[
	m(X,\beta,h_{1},h_{2})=w(X)(h_{1}(X)-h_{2}(X))-\beta=D(X)^{\prime}h(X)-\beta,
	\]
	with $D(X)\equiv w(X)(1,-1)^{\prime}$ and $h(X)\equiv(h_{1}(X),h_{2}%
	(X))^{\prime}$. We use \citet[equation 4.1]{Newey_1994}'s notation $D(X)$ for the functional derivative which is not the treatment
	dummy $D$.
	
	Following \citet{Newey_1994}, define a path indexed by the scalar parameter
	$\theta$ for the distribution of $(Y,D,X)$ with density $f^{\ast}(\cdot
	,\theta)$, where $f^{\ast}(\cdot,0)=f^{\ast}(\cdot)$ is the density of the
	distribution of $(Y,D,X)$ in the stratified sample. If $\mathbb{E}_{\theta
	}^{\ast}$ denotes an expectation with respect to the distribution with density
	$f^{\ast}(\cdot,\theta)$, then we define the corresponding paths for the
	projections $h_{1}(X,\theta)\equiv\mathbb{E}_{\theta}^{\ast}[Y|D=1,X]$ and
	$h_{2}(X,\theta)\equiv\mathbb{E}_{\theta}^{\ast}[Y|D=0,X]$. The paths
	$h_{1}(X,\theta)$ and $h_{2}(X,\theta)$ are the minimizers of the objective
	functions
	\[
	\mathbb{E}_{\theta}^{\ast}[D(Y-\tilde{h}_{1}(X,\theta))^{2}]\quad
	\text{and}\quad\mathbb{E}_{\theta}^{\ast}[(1-D)(Y-\tilde{h}_{2}(X,\theta
	))^{2}],
	\]
	respectively, so the following orthogonality conditions hold for all functions
	$\tilde{h}_{1}(X,\theta)$ and $\tilde{h}_{2}(X,\theta)$:
	\[
	\mathbb{E}_{\theta}^{\ast}[D(Y-h_{1}(X,\theta))\tilde{h}_{1}(X,\theta
	)]=0\quad\text{and}\quad\mathbb{E}_{\theta}^{\ast}[(1-D)(Y-h_{2}
	(X,\theta))\tilde{h}_{2}(X,\theta)]=0.
	\]
	We sum up the orthogonality conditions
	\[
	\mathbb{E}_{\theta}^{\ast}[D(Y-h_{1}(X,\theta))\tilde{h}_{1}(X,\theta
	)+(1-D)(Y-h_{2}(X,\theta))\tilde{h}_{2}(X,\theta)]=0
	\]
	for all functions $(\tilde{h}_{1}(X,\theta),\tilde{h}_{2}(X,\theta))$. Choose
	$(\tilde{h}_{1}(X,\theta),\tilde{h}_{2}(X,\theta))=\left(  \frac{w(X)}
	{\pi^{\ast}(X,\theta)},-\frac{w(X)}{1-\pi^{\ast}(X,\theta)}\right)  $, where
	$\pi^{\ast}(X,\theta)\equiv\mathbb{E}_{\theta}^{\ast}[D|X]$. We obtain
	\begin{equation}
		\mathbb{E}_{\theta}^{\ast}\left[  w(X)\left(  \frac{D}{\pi^{\ast}(X,\theta
			)}(Y-h_{1}(X,\theta))-\frac{(1-D)}{1-\pi^{\ast}(X,\theta)}(Y-h_{2}
		(X,\theta))\right)  \right]  =0, \label{key1}%
	\end{equation}
	which is equivalent to
	\begin{align}
		&  \mathbb{E}_{\theta}^{\ast}\left[  w(X)\left(  \frac{D}{\pi^{\ast}
			(X,\theta)}Y-\frac{1-D}{1-\pi^{\ast}(X,\theta)}Y\right)  \right] \nonumber\\
		&  =\mathbb{E}_{\theta}^{\ast}\left[  w(X)\left(  \frac{D}{\pi^{\ast}
			(X,\theta)}h_{1}(X,\theta)-\frac{1-D}{1-\pi^{\ast}(X,\theta)}h_{2}
		(X,\theta)\right)  \right] \nonumber\\
		&  =\mathbb{E}_{\theta}^{\ast}[w(X)(h_{1}(X,\theta)-h_{2}(X,\theta))].
		\label{key2}%
	\end{align}
	As in \citet[equation 4.5]{Newey_1994} we compute the total derivative of (\ref{key2}) at $\theta=0$
	\begin{align}
		\frac{\partial\mathbb{E}_{\theta}^{\ast}[w(X)(h_{1}(X,\theta)-h_{2}
			(X,\theta))]}{\partial\theta}=  &  \frac{\partial\mathbb{E}_{\theta}^{\ast
			}[w(X)(h_{1}(X)-h_{2}(X))]}{\partial\theta}\nonumber\\
		&  +\frac{\partial\mathbb{E}^{\ast}[w(X)(h_{1}(X,\theta)-h_{2}(X,\theta
			))]}{\partial\theta}, \label{chain-rule}%
	\end{align}
	where we used the fact that the projections at $\theta=0$ are equal to
	$h_{1}(X)$ and $h_{2}(X)$. Therefore combining (\ref{chain-rule}) with
	(\ref{key2}), we obtain
	\begin{align*}
		\frac{\partial\mathbb{E}^{\ast}[D(X)^{\prime}h(X,\theta)]}{\partial\theta}=
		&  \frac{\partial\mathbb{E}^{\ast}[w(X)(h_{1}(X,\theta)-h_{2}(X,\theta
			))]}{\partial\theta}\\
		=  &  \frac{\partial\mathbb{E}_{\theta}^{\ast}[w(X)(h_{1}(X,\theta
			)-h_{2}(X,\theta))]}{\partial\theta}-\frac{\partial\mathbb{E}_{\theta}^{\ast
			}[w(X)(h_{1}(X)-h_{2}(X))]}{\partial\theta}\\
		=  &  \frac{\partial}{\partial\theta}\left(  \mathbb{E}_{\theta}^{\ast}\left[
		w(X)\left(  \frac{D}{\pi^{\ast}(X,\theta)}Y-\frac{1-D}{1-\pi^{\ast}(X,\theta
			)}Y\right)  \right]  \right) \\
		&  -\frac{\partial}{\partial\theta}\left(  \mathbb{E}_{\theta}^{\ast}\left[
		w(X)\left(  \frac{D}{\pi^{\ast}(X,\theta)}h_{1}(X)-\frac{1-D}{1-\pi^{\ast
			}(X,\theta)}h_{2}(X)\right)  \right]  \right) \\
		=  &  \frac{\partial}{\partial\theta}\left(  \mathbb{E}_{\theta}^{\ast}\left[
		w(X)\left(  \frac{D}{\pi^{\ast}(X,\theta)}\left(  Y-h_{1}(X)\right)
		-\frac{1-D}{1-\pi^{\ast}(X,\theta)}\left(  Y-h_{2}(X)\right)  \right)
		\right]  \right)  .
	\end{align*}
	By taking the total derivative of the right-hand side at $\theta=0$ and noting
	that the second term after the first equality is 0, we derive
	\begin{align*}
		\frac{\partial\mathbb{E}^{\ast}[D(X)^{\prime}h(X,\theta)]}{\partial\theta}=
		&  \mathbb{E}^{\ast}\left[  w(X)\left(  \frac{D}{\pi^{\ast}(X)}(Y-h_{1}
		(X))-\frac{1-D}{1-\pi^{\ast}(X)}(Y-h_{2}(X))\right)  S(Y,D,X)\right] \\
		&  +\mathbb{E}^{\ast}\left[  w(X)\left(  -\frac{\dot{\pi}^{\ast}(X)D}
		{\pi^{\ast}(X)^{2}}(Y-h_{1}(X))-\frac{\dot{\pi}^{\ast}(X)(1-D)}{(1-\pi^{\ast
			}(X))^{2}}(Y-h_{2}(X))\right)  \right] \\
		=  &  \mathbb{E}^{\ast}\left[  w(X)\left(  \frac{D}{\pi^{\ast}(X)}
		(Y-h_{1}(X))-\frac{1-D}{1-\pi^{\ast}(X)}(Y-h_{2}(X))\right)  S(Y,D,X)\right]
		.
	\end{align*}
	In this expression $S(\cdot)\equiv\left.  \left.  \partial\ln f^{\ast}
	(\cdot,\theta)\right/  \partial\theta\right\vert _{\theta=0}$ is the score of
	the density, and $\dot{\pi}^{\ast}(X)\equiv\left.  \left.  \partial\pi^{\ast
	}(X,\theta)\right/  \partial\theta\right\vert _{\theta=0}$. Therefore by \citet[Proposition 4]{Newey_1994} the adjustment to the influence function is
	\[
	w(X)\left(  \frac{D}{\pi^{\ast}(X)}\left(  Y-h_{1}(X)\right)  -\frac
	{1-D}{1-\pi^{\ast}(X)}\left(  Y-h_{2}(X)\right)  \right)  .
	\]
	
\end{proof}

\begin{proof}
	[Proof of Lemma \ref{lem:ATE-PS-adjust}]Adopting \cite{HR_2013}'s notation, we
	let
	\[
	h(D,\mu_{1},\mu_{2})\equiv\left(  D\frac{\pi}{\pi^{\ast}}+(1-D)\frac{1-\pi
	}{1-\pi^{\ast}}\right)  (\mu_{1}(\pi^{\ast}(X))-\mu_{2}(\pi^{\ast}(X))),
	\]
	where $\mu_{1}(v)\equiv E^{\ast}[Y|D=1,\pi^{\ast}(X)=v]=\delta_{1}^{\ast}(v)$
	and $\mu_{2}(v)\equiv E^{\ast}[Y|D=0,\pi^{\ast}(X)=v]=\delta_{0}^{\ast}(v)$.
	Also $E^{\ast}[D|\pi^{\ast}(X)]=\pi^{\ast}(X)$,
	\begin{align*}
		\frac{\partial h(D,\mu_{1},\mu_{2})}{\partial\mu_{1}}  &  =D\frac{\pi}%
		{\pi^{\ast}}+(1-D)\frac{1-\pi}{1-\pi^{\ast}},\\
		\frac{\partial h(D,\mu_{1},\mu_{2})}{\partial\mu_{2}}  &  =-\left(  D\frac
		{\pi}{\pi^{\ast}}+(1-D)\frac{1-\pi}{1-\pi^{\ast}}\right)  ,\\
		\kappa_{1}(v)  &  \equiv E^{\ast}\left[  \left.  \frac{\partial h(D,\mu
			_{1},\mu_{2})}{\partial\mu_{1}}\right\vert \pi^{\ast}(X)=v\right]  =v\frac
		{\pi}{\pi^{\ast}}+(1-v)\frac{1-\pi}{1-\pi^{\ast}},\\
		\kappa_{2}(v)  &  \equiv E^{\ast}\left[  \left.  \frac{\partial h(D,\mu
			_{1},\mu_{2})}{\partial\mu_{2}}\right\vert \pi^{\ast}(X)=v\right]
		=-v\frac{\pi}{\pi^{\ast}}-(1-v)\frac{1-\pi}{1-\pi^{\ast}},\\
		\frac{\partial\kappa_{1}(v)}{\partial v}  &  =\frac{\pi}{\pi^{\ast}}%
		-\frac{1-\pi}{1-\pi^{\ast}},\\
		\frac{\partial\kappa_{2}(v)}{\partial v}  &  =-\frac{\pi}{\pi^{\ast}}%
		+\frac{1-\pi}{1-\pi^{\ast}},
	\end{align*}
	and finally $\pi_{1}(v)\equiv v$, $\pi_{2}(v)\equiv1-v$. Further, we denote
	$\mu_{1}(X)\equiv E^{\ast}[Y|D=1,X]=\beta_{1}(X)$ and $\mu_{2}(X)\equiv
	E^{\ast}[Y|D=0,X]=\beta_{0}(X)$. Using \citet[Theorem 7]{HR_2013}, we can see that the contribution of estimating $\pi^{\ast}(X)$ is the sum of
	the three terms in (\ref{HR1}), (\ref{HR2}), and (\ref{HR3}) times
	$D-\pi^{\ast}(X)$:
	\begin{align}
		E^{\ast}\left[  \left.  \left(  \frac{\partial h(D,\mu_{1},\mu_{2})}%
		{\partial\mu_{1}}-\kappa_{1}(\pi^{\ast}(X))\right)  \frac{\partial\mu_{1}%
			(\pi^{\ast}(X))}{\partial v}\right\vert X\right]   & \nonumber\\
		+E^{\ast}\left[  \left.  \left(  \frac{\partial h(D,\mu_{1},\mu_{2})}%
		{\partial\mu_{2}}-\kappa_{2}(\pi^{\ast}(X))\right)  \frac{\partial\mu_{2}%
			(\pi^{\ast}(X))}{\partial v}\right\vert X\right]   &  =0, \label{HR1}%
	\end{align}%
	\begin{align}
		&  E^{\ast}\left[  \left.  (\mu_{1}(X)-\mu_{1}(\pi^{\ast}(X)))\frac
		{\partial\kappa_{1}(\pi^{\ast}(X))}{\partial v}\right\vert X\right]  +E^{\ast
		}\left[  \left.  (\mu_{2}(X)-\mu_{2}(\pi^{\ast}(X)))\frac{\partial\kappa
			_{2}(\pi^{\ast}(X))}{\partial v}\right\vert X\right] \nonumber\\
		=  &  (\beta_{1}(X)-\delta_{1}^{\ast}(\pi^{\ast}(X)))\left(  \frac{\pi}%
		{\pi^{\ast}}-\frac{1-\pi}{1-\pi^{\ast}}\right)  -(\beta_{0}(X)-\delta
		_{0}^{\ast}(\pi^{\ast}(X)))\left(  \frac{\pi}{\pi^{\ast}}-\frac{1-\pi}%
		{1-\pi^{\ast}}\right) \nonumber\\
		=  &  (\beta(X)-\delta^{\ast}(\pi^{\ast}(X)))\left(  \frac{\pi}{\pi^{\ast}%
		}-\frac{1-\pi}{1-\pi^{\ast}}\right)  , \label{HR2}%
	\end{align}
	and
	\begin{align}
		&  -E^{\ast}\left[  \left.  \frac{1}{\pi_{1}(\pi^{\ast}(X))}(\mu_{1}%
		(X)-\mu_{1}(\pi^{\ast}(X)))\kappa_{1}(\pi^{\ast}(X))\frac{\partial\pi_{1}%
			(\pi^{\ast}(X))}{\partial v}\right\vert X\right] \nonumber\\
		&  -E^{\ast}\left[  \left.  \frac{1}{\pi_{2}(\pi^{\ast}(X))}(\mu_{2}%
		(X)-\mu_{2}(\pi^{\ast}(X)))\kappa_{2}(\pi^{\ast}(X))\frac{\partial\pi_{2}%
			(\pi^{\ast}(X))}{\partial v}\right\vert X\right] \nonumber\\
		=  &  -\frac{1}{\pi^{\ast}(X)}(\beta_{1}(X)-\delta_{1}^{\ast}(\pi^{\ast
		}(X)))\left(  \pi^{\ast}(X)\frac{\pi}{\pi^{\ast}}+(1-\pi^{\ast}(X))\frac
		{1-\pi}{1-\pi^{\ast}}\right) \nonumber\\
		&  -\frac{1}{1-\pi^{\ast}(X)}(\beta_{0}(X)-\delta_{0}^{\ast}(\pi^{\ast
		}(X)))\left(  -\pi^{\ast}(X)\frac{\pi}{\pi^{\ast}}-(1-\pi^{\ast}%
		(X))\frac{1-\pi}{1-\pi^{\ast}}\right)  (-1)\nonumber\\
		=  &  -\frac{1}{\pi^{\ast}(X)}(\beta_{1}(X)-\delta_{1}^{\ast}(\pi^{\ast
		}(X)))\left(  \pi^{\ast}(X)\frac{\pi}{\pi^{\ast}}+(1-\pi^{\ast}(X))\frac
		{1-\pi}{1-\pi^{\ast}}\right) \nonumber\\
		&  -\frac{1}{1-\pi^{\ast}(X)}(\beta_{0}(X)-\delta_{0}^{\ast}(\pi^{\ast
		}(X)))\left(  \pi^{\ast}(X)\frac{\pi}{\pi^{\ast}}+(1-\pi^{\ast}(X))\frac
		{1-\pi}{1-\pi^{\ast}}\right)  . \label{HR3}%
	\end{align}
	Combining them, we obtain
	\begin{align}
		&  (\beta(X)-\delta(\pi^{\ast}(X)))\left(  \frac{\pi}{\pi^{\ast}}-\frac{1-\pi
		}{1-\pi^{\ast}}\right)  (D-\pi^{\ast}(X))\nonumber\\
		&  -\frac{1}{\pi^{\ast}(X)}(\beta_{1}(X)-\delta_{1}(\pi^{\ast}(X)))\left(
		\pi^{\ast}(X)\frac{\pi}{\pi^{\ast}}+(1-\pi^{\ast}(X))\frac{1-\pi}{1-\pi^{\ast
		}}\right)  (D-\pi^{\ast}(X))\nonumber\\
		&  -\frac{1}{1-\pi^{\ast}(X)}(\beta_{0}(X)-\delta_{0}(\pi^{\ast}(X)))\left(
		\pi^{\ast}(X)\frac{\pi}{\pi^{\ast}}+(1-\pi^{\ast}(X))\frac{1-\pi}{1-\pi^{\ast
		}}\right)  (D-\pi^{\ast}(X)) \label{HR-all}%
	\end{align}
	as the adjustment for the estimation of $\pi^{\ast}(X)$. We can rewrite
	(\ref{HR-all}) as
	\begin{align*}
		&  (\beta(X)-\delta^{\ast}(\pi^{\ast}(X)))\left(  D\frac{\pi}{\pi^{\ast}%
		}+(1-D)\frac{1-\pi}{1-\pi^{\ast}}-\left(  \pi^{\ast}(X)\frac{\pi}{\pi^{\ast}%
		}+(1-\pi^{\ast}(X))\frac{1-\pi}{1-\pi^{\ast}}\right)  \right) \\
		&  -(\beta_{1}(X)-\delta_{1}(\pi^{\ast}(X)))\left(  \pi^{\ast}(X)\frac{\pi
		}{\pi^{\ast}}+(1-\pi^{\ast}(X))\frac{1-\pi}{1-\pi^{\ast}}\right)  \left(
		\frac{D}{\pi^{\ast}(X)}-1\right) \\
		&  -(\beta_{0}(X)-\delta_{0}(\pi^{\ast}(X)))\left(  \pi^{\ast}(X)\frac{\pi
		}{\pi^{\ast}}+(1-\pi^{\ast}(X))\frac{1-\pi}{1-\pi^{\ast}}\right)  \left(
		1-\frac{1-D}{1-\pi^{\ast}(X)}\right) \\
		=  &  \left(  D\frac{\pi}{\pi^{\ast}}+(1-D)\frac{1-\pi}{1-\pi^{\ast}}\right)
		(\beta(X)-\delta^{\ast}(\pi^{\ast}(X)))\\
		&  -\frac{\pi^{\ast}(X)\frac{\pi}{\pi^{\ast}}+(1-\pi^{\ast}(X))\frac{1-\pi
			}{1-\pi^{\ast}}}{\pi^{\ast}(X)}D(\beta_{1}(X)-\delta_{1}(\pi^{\ast}(X)))\\
		&  +\frac{\pi^{\ast}(X)\frac{\pi}{\pi^{\ast}}+(1-\pi^{\ast}(X))\frac{1-\pi
			}{1-\pi^{\ast}}}{1-\pi^{\ast}(X)}(1-D)(\beta_{0}(X)-\delta_{0}(\pi^{\ast
		}(X))),
	\end{align*}
	from which we obtain the conclusion.
\end{proof}

\begin{proof}
	[Proof of Lemma \ref{lem:ATT-adjust}]We follow the proof of Lemma
	\ref{lem:ATE-adjust}, using similar notation. We write the numerator in
	(\ref{ATT-estimator}) as
	\begin{equation}
		\frac{1}{n}\sum_{i=1}^{n}D_{i}(Y_{i}-\hat{\beta}_{0}(X_{i}))=\frac{1}{n}%
		\sum_{i=1}^{n}D_{i}Y_{i}-\frac{1}{n}\sum_{i=1}^{n}D_{i}\hat{\beta}_{0}(X_{i}),
		\label{ATT-numerator}%
	\end{equation}
	and focus on $\frac{1}{n}\sum_{i=1}^{n}D_{i}\hat{\beta}_{0}(X_{i})$, which
	estimates $\beta^{\ast}\equiv\mathbb{E}^{\ast}[\pi^{\ast}(X)h_{2}(X)]$, where
	$h_{2}(X)\equiv\mathbb{E}^{\ast}[Y|D=0,X]=\beta_{0}(X)$. So we consider the
	moment equation $\mathbb{E}^{\ast}[m(X,\beta^{\ast},h_{2})]=0$, where
	\[
	m(X,\beta^{\ast},h_{2})=\pi^{\ast}(X)h_{2}(X)-\beta^{\ast}.
	\]
	In \citet[equation 4.1]{Newey_1994}'s notation, we therefore want to consider $\mathbb{E}^{\ast}[D(X)h_{2}(X)]$
	with $D(X)\equiv\pi^{\ast}(X)$, and $D(X)h_{2}(X)$ is obviously linear in
	$h_{2}$. Recall that the following orthogonality condition holds for all
	functions $\tilde{h}_{2}(X,\theta)$:
	\[
	\mathbb{E}_{\theta}^{\ast}[(1-D)(Y-h_{2}(X,\theta))\tilde{h}_{2}%
	(X,\theta)]=0.
	\]
	Choose $\tilde{h}_{2}(X,\theta)=\frac{\pi^{\ast}(X)}{1-\pi^{\ast}(X,\theta)}$,
	where $\pi^{\ast}(X,\theta)\equiv\mathbb{E}_{\theta}^{\ast}[D|X]$, so that
	\[
	0=\mathbb{E}_{\theta}^{\ast}\left[  \frac{\pi^{\ast}(X)(1-D)}{1-\pi^{\ast
		}(X,\theta)}(Y-h_{2}(X,\theta))\right]
	\]
	or
	\[
	\mathbb{E}_{\theta}^{\ast}\left[  \frac{\pi^{\ast}(X)(1-D)}{1-\pi^{\ast
		}(X,\theta)}Y\right]  =\mathbb{E}_{\theta}^{\ast}\left[  \frac{\pi^{\ast
		}(X)(1-D)}{1-\pi^{\ast}(X,\theta)}h_{2}(X,\theta)\right]  =\mathbb{E}_{\theta
	}^{\ast}[\pi^{\ast}(X)h_{2}(X,\theta)].
	\]
	The last expression is useful to compute the derivative as in \citet[equation 4.5]{Newey_1994}. Taking the total derivative at $\theta=0$%
	\begin{equation}
		\frac{\partial\mathbb{E}_{\theta}^{\ast}[\pi^{\ast}(X)h_{2}(X,\theta
			)]}{\partial\theta}=\frac{\partial\mathbb{E}_{\theta}^{\ast}[\pi^{\ast
			}(X)h_{2}(X)]}{\partial\theta}+\frac{\partial\mathbb{E}^{\ast}[\pi^{\ast
			}(X)h_{2}(X,\theta)]}{\partial\theta},\nonumber
	\end{equation}
	we obtain
	\begin{align*}
		\frac{\partial\mathbb{E}^{\ast}[\pi^{\ast}(X)h_{2}(X,\theta)]}{\partial\theta
		}=  &  \frac{\partial\mathbb{E}_{\theta}^{\ast}[\pi^{\ast}(X)h_{2}(X,\theta
			)]}{\partial\theta}-\frac{\partial\mathbb{E}_{\theta}^{\ast}[\pi^{\ast
			}(X)h_{2}(X)]}{\partial\theta}\\
		=  &  \frac{\partial}{\partial\theta}\mathbb{E}_{\theta}^{\ast}\left[
		\frac{\pi^{\ast}(X)(1-D)}{1-\pi^{\ast}(X,\theta)}Y\right]  -\frac{\partial
		}{\partial\theta}\mathbb{E}_{\theta}^{\ast}\left[  \frac{\pi^{\ast}%
			(X)(1-D)}{1-\pi^{\ast}(X,\theta)}h_{2}(X)\right] \\
		=  &  \frac{\partial}{\partial\theta}\mathbb{E}_{\theta}^{\ast}\left[
		\frac{\pi^{\ast}(X)(1-D)}{1-\pi^{\ast}(X,\theta)}(Y-h_{2}(X))\right] \\
		=  &  \mathbb{E}^{\ast}\left[  \frac{\pi^{\ast}(X)(1-D)}{1-\pi^{\ast}%
			(X)}(Y-h_{2}(X))S(Y,D,X)\right] \\
		&  +\mathbb{E}^{\ast}\left[  \pi^{\ast}(X)(1-D)\frac{\dot{\pi}^{\ast}%
			(X)}{(1-\pi^{\ast}(X))^{2}}(Y-h_{2}(X))\right] \\
		=  &  \mathbb{E}^{\ast}\left[  \frac{\pi^{\ast}(X)(1-D)}{1-\pi^{\ast}%
			(X)}(Y-h_{2}(X))S(Y,D,X)\right]  ,
	\end{align*}
	so the adjustment for $\frac{1}{n}\sum_{i=1}^{n}D_{i}\hat{\beta}_{0}(X_{i})$
	is
	\[
	\frac{\pi^{\ast}(X)(1-D)}{1-\pi^{\ast}(X)}(Y-h_{2}(X))=\frac{\pi^{\ast
		}(X)(1-D)}{1-\pi^{\ast}(X)}(Y-\beta_{0}(X)),
	\]
	from which we obtain the conclusion.
\end{proof}

\begin{proof}
	[Proof of Lemma \ref{lem:ATT-PS-adjust}]We adopt \cite{HR_2013}'s notation
	similarly as in Lemma \ref{lem:ATE-PS-adjust}, except that we let
	$h(D,Y,\mu_{2})\equiv D(Y-\mu_{2}(\pi^{\ast}(X)))$ and
	\begin{align*}
		\frac{\partial h(D,Y,\mu_{2})}{\partial\mu_{2}}  &  =-D,\\
		\kappa_{2}(v)  &  \equiv E^{\ast}\left[  \left.  \frac{\partial h(D,Y,\mu
			_{2})}{\partial\mu_{2}}\right\vert \pi^{\ast}(X)=v\right]  =-v,\\
		\frac{\partial\kappa_{2}(v)}{\partial v}  &  =-1.
	\end{align*}
	Applying \citet[Theorem 7]{HR_2013},\footnote{\citet[p.333]{HR_2013} considered an imputation version, not this version of the estimated propensity
		score matching for the ATT.} we can see that the adjustment for the estimation
	of $\pi^{\ast}(X)$ is the sum of the three terms in (\ref{ATT-HR1}),
	(\ref{ATT-HR2}), and (\ref{ATT-HR3}) times $D-\pi^{\ast}(X)$:
	\begin{equation}
		E^{\ast}\left[  \left.  \left(  \frac{\partial h(D,Y,\mu_{2})}{\partial\mu
			_{2}}-\kappa_{2}(\pi^{\ast}(X))\right)  \frac{\partial\mu_{2}(\pi^{\ast}%
			(X))}{\partial v}\right\vert X\right]  =0, \label{ATT-HR1}%
	\end{equation}%
	\begin{equation}
		E^{\ast}\left[  \left.  (\mu_{2}(X)-\mu_{2}(\pi^{\ast}(X)))\frac
		{\partial\kappa_{2}(\pi^{\ast}(X))}{\partial v}\right\vert X\right]
		=-(\beta_{0}(X)-\delta_{0}^{\ast}(\pi^{\ast}(X))), \label{ATT-HR2}%
	\end{equation}
	and
	\begin{align}
		&  -E^{\ast}\left[  \left.  \frac{1}{1-\pi^{\ast}(X)}(\mu_{2}(X)-\mu_{2}%
		(\pi^{\ast}(X)))\kappa_{2}(\pi^{\ast}(X))\frac{\partial\pi_{2}(\pi^{\ast}%
			(X))}{\partial v}\right\vert X\right] \nonumber\\
		&  =-\frac{1}{1-\pi^{\ast}(X)}(\beta_{0}(X)-\delta_{0}^{\ast}(\pi^{\ast
		}(X)))\pi^{\ast}(X). \label{ATT-HR3}%
	\end{align}
	Combining them, we obtain the desired conclusion.
\end{proof}

\subsection{Proofs of the Propositions}

\begin{proof}
	[Proof of Proposition \ref{prop:ATE}]We derive the asymptotic distribution of
	the ATE estimator based on the moments (\ref{ATE-estimator}). The estimator
	for the ATE is
	\[
	\hat{\beta}=\frac{1}{n}\sum_{i=1}^{n}\left(  D_{i}\frac{\pi}{\pi^{\ast}%
	}+(1-D_{i})\frac{1-\pi}{1-\pi^{\ast}}\right)  (\hat{\beta}_{1}(X_{i}%
	)-\hat{\beta}_{0}(X_{i})).
	\]
	Standard arguments (e.g., \citealp{Newey_1994}) can be used to show that the
	influence function of $\hat{\beta}$ is equal to a main term
	\[
	\left(  D\frac{\pi}{\pi^{\ast}}+(1-D)\frac{1-\pi}{1-\pi^{\ast}}\right)
	(\beta_{1}(X)-\beta_{0}(X))-\beta
	\]
	plus the adjustment term for the estimation of $(\beta_{1}(X),\beta_{0}(X))$.
	Lemma \ref{lem:ATE-adjust} shows that this adjustment term is given by
	(\ref{ATE-adjust}), so the influence function is equal to
	\begin{align}
		&  \left(  D\frac{\pi}{\pi^{\ast}}+(1-D)\frac{1-\pi}{1-\pi^{\ast}}\right)
		(\beta_{1}(X)-\beta_{0}(X))-\beta\nonumber\\
		&  +\left(  \pi^{\ast}(X)\frac{\pi}{\pi^{\ast}}+(1-\pi^{\ast}(X))\frac{1-\pi
		}{1-\pi^{\ast}}\right)  \left(  \frac{D}{\pi^{\ast}(X)}(Y-\beta_{1}%
		(X))-\frac{1-D}{1-\pi^{\ast}(X)}(Y-\beta_{0}(X))\right)  .
		\label{matchX-influence}%
	\end{align}
	If $\pi^{\ast}$ is estimated, we need to add the following term to
	(\ref{matchX-influence})
	\begin{align}
		&  E^{\ast}\left[  \frac{\partial}{\partial\pi^{\ast}}\left(  D\frac{\pi}%
		{\pi^{\ast}}+(1-D)\frac{1-\pi}{1-\pi^{\ast}}\right)  (\beta_{1}(X)-\beta
		_{0}(X))\right]  (D-\pi^{\ast})\nonumber\\
		&  =E^{\ast}\left[  \left(  -D\frac{\pi}{\left(  \pi^{\ast}\right)  ^{2}%
		}+(1-D)\frac{1-\pi}{(1-\pi^{\ast})^{2}}\right)  (\beta_{1}(X)-\beta
		_{0}(X))\right]  (D-\pi^{\ast})\nonumber\\
		&  =E^{\ast}\left[  \left(  -\pi^{\ast}(X)\frac{\pi}{(\pi^{\ast})^{2}}%
		+(1-\pi^{\ast}(X))\frac{1-\pi}{(1-\pi^{\ast})^{2}}\right)  (\beta_{1}%
		(X)-\beta_{0}(X))\right]  (D-\pi^{\ast}) \label{matchX-influence-estimatedpi*}%
	\end{align}
	
\end{proof}

\begin{proof}
	[Proof of Proposition \ref{prop:ATE-PS}]\label{pf:propPS}The ATE estimator
	based on the moments (\ref{ATE-PS-estimator}) takes the form
	\[
	\frac{1}{n}\sum_{i=1}^{n}\left(  D_{i}\frac{\pi}{\pi^{\ast}}+(1-D_{i}%
	)\frac{1-\pi}{1-\pi^{\ast}}\right)  (\hat{\delta}_{1}(\hat{\pi}(X_{i}%
	))-\hat{\delta}_{0}(\hat{\pi}(X_{i}))).
	\]
	The influence function of this estimator has as its main term
	\begin{equation}
		\left(  D\frac{\pi}{\pi^{\ast}}+(1-D)\frac{1-\pi}{1-\pi^{\ast}}\right)
		(\delta_{1}^{\ast}(\pi^{\ast}(X))-\delta_{0}^{\ast}(\pi^{\ast}(X)))-\beta,
		\label{ATE-PS-main}%
	\end{equation}
	to which we add the adjustment for the estimation of $(\delta_{1}^{\ast}%
	(\cdot),\delta_{0}^{\ast}(\cdot))$ and the adjustment for the estimation of
	$\pi^{\ast}(\cdot)$. The first adjustment can be derived following Lemma
	\ref{lem:ATE-adjust} as
	\begin{equation}
		\left(  \pi^{\ast}(X)\frac{\pi}{\pi^{\ast}}+(1-\pi^{\ast}(X))\frac{1-\pi
		}{1-\pi^{\ast}}\right)  \left(  \frac{D}{\pi^{\ast}(X)}(Y-\delta_{1}^{\ast
		}(\pi^{\ast}(X)))-\frac{1-D}{1-\pi^{\ast}(X)}(Y-\delta_{0}^{\ast}(\pi^{\ast
		}(X)))\right)  . \label{ATE-PS-beta}%
	\end{equation}
	Lemma \ref{lem:ATE-PS-adjust} shows that the second adjustment is given by
	(\ref{ATE-PS-pi}).
	
	Adding up (\ref{ATE-PS-main}) and (\ref{ATE-PS-beta}), we can see that the
	influence function of the infeasible estimator based on known $\pi^{\ast}(X)$
	is equal to
	\begin{align}
		&  \left(  \left(  D\frac{\pi}{\pi^{\ast}}+(1-D)\frac{1-\pi}{1-\pi^{\ast}%
		}\right)  \delta^{\ast}(\pi^{\ast}(X))-\beta\right) \nonumber\\
		&  +\frac{\pi^{\ast}(X)\frac{\pi}{\pi^{\ast}}+(1-\pi^{\ast}(X))\frac{1-\pi
			}{1-\pi^{\ast}}}{\pi^{\ast}(X)}D\tilde{\varepsilon}_{1}\nonumber\\
		&  -\frac{\pi^{\ast}(X)\frac{\pi}{\pi^{\ast}}+(1-\pi^{\ast}(X))\frac{1-\pi
			}{1-\pi^{\ast}}}{1-\pi^{\ast}(X)}(1-D)\tilde{\varepsilon}_{0}, \label{ATE1}%
	\end{align}
	where $\tilde{\varepsilon}_{d}\equiv Y-\delta_{d}^{\ast}(\pi^{\ast}(X))$.
	Because\ $\tilde{\varepsilon}_{d}=\varepsilon_{d}+\beta_{d}(X)-\delta
	_{d}^{\ast}(\pi^{\ast}(X))$, the sum of (\ref{ATE1}) and (\ref{ATE-PS-pi}) is
	\begin{align*}
		&  \left(  \left(  D\frac{\pi}{\pi^{\ast}}+(1-D)\frac{1-\pi}{1-\pi^{\ast}%
		}\right)  \beta(X)-\beta\right) \\
		&  +\frac{\pi^{\ast}(X)\frac{\pi}{\pi^{\ast}}+(1-\pi^{\ast}(X))\frac{1-\pi
			}{1-\pi^{\ast}}}{\pi^{\ast}(X)}D\varepsilon_{1}\\
		&  -\frac{\pi^{\ast}(X)\frac{\pi}{\pi^{\ast}}+(1-\pi^{\ast}(X))\frac{1-\pi
			}{1-\pi^{\ast}}}{1-\pi^{\ast}(X)}(1-D)\varepsilon_{0}.
	\end{align*}
	This is identical to the influence function (\ref{ATE-infl}) of the ATE
	estimator that conditions on $X$.
	
	If $\pi^{\ast}$ is estimated, the influence function needs to reflect it by
	adding
	\begin{align*}
		&  E^{\ast}\left[  \frac{\partial}{\partial\pi^{\ast}}\left(  D\frac{\pi}%
		{\pi^{\ast}}+(1-D)\frac{1-\pi}{1-\pi^{\ast}}\right)  (\delta_{1}^{\ast}%
		(\pi^{\ast}(X))-\delta_{0}^{\ast}(\pi^{\ast}(X)))\right]  (D-\pi^{\ast})\\
		&  =E^{\ast}\left[  \left(  -\pi^{\ast}(X)\frac{\pi}{(\pi^{\ast})^{2}}%
		+(1-\pi^{\ast}(X))\frac{1-\pi}{(1-\pi^{\ast})^{2}}\right)  (\delta_{1}^{\ast
		}(\pi^{\ast}(X))-\delta_{0}^{\ast}(\pi^{\ast}(X)))\right]  (D-\pi^{\ast}),
	\end{align*}
	which is identical to (\ref{matchX-influence-estimatedpi*}).
\end{proof}

\begin{proof}
	[Proof of Proposition \ref{prop:ATT}]Now we derive the asymptotic variance of
	the ATT estimator in (\ref{ATT-estimator}), where $\hat{\beta}_{0}(X)$ is a
	nonparametric regression of $Y$ on $X$ in the control group ($D=0$). Lemma
	\ref{lem:ATT-adjust} establishes the adjustment for the estimation of
	$\beta_{0}(X)$ in the numerator of (\ref{ATT-estimator}) as
	\[
	-\frac{\pi^{\ast}(X)}{1-\pi^{\ast}(X)}(1-D)(Y-\beta_{0}(X)).
	\]
	It follows that the numerator $\frac{1}{n}\sum_{i=1}^{n}D_{i}(Y_{i}-\hat
	{\beta}_{0}(X_{i}))$ has the influence function
	\begin{align}
		&  D(Y-\beta_{0}(X))-E^{\ast}[\pi^{\ast}(X)\beta(X)]-\frac{\pi^{\ast}%
			(X)}{1-\pi^{\ast}(X)}(1-D)(Y-\beta_{0}(X))\nonumber\\
		=  &  D(\beta_{1}(X)+\varepsilon_{1}-\beta_{0}(X))-E^{\ast}[\pi^{\ast}%
		(X)\beta(X)]-\frac{\pi^{\ast}(X)}{1-\pi^{\ast}(X)}(1-D)\varepsilon
		_{0}\nonumber\\
		=  &  D\beta(X)-E^{\ast}[\pi^{\ast}(X)\beta(X)]+D\varepsilon_{1}-\frac
		{\pi^{\ast}(X)}{1-\pi^{\ast}(X)}(1-D)\varepsilon_{0}.
		\label{ATT-num-influence}%
	\end{align}
	In other words,
	\begin{align*}
		&  \frac{1}{n}\sum_{i=1}^{n}(D_{i}(Y_{i}-\hat{\beta}_{0}(X_{i}))-E^{\ast}%
		[\pi^{\ast}(X)\beta(X)])\\
		=  &  \frac{1}{n}\sum_{i=1}^{n}\left(  D_{i}\beta(X_{i})-E^{\ast}[\pi^{\ast
		}(X)\beta(X)]+D_{i}\varepsilon_{1i}-\frac{\pi(X_{i})}{1-\pi(X_{i})}%
		(1-D_{i})\varepsilon_{0i}\right)  +o_{p}\left(  \frac{1}{\sqrt{n}}\right)  .
	\end{align*}
	By the delta method, we obtain
	\begin{align*}
		&  \sqrt{n}\left(  \frac{\frac{1}{n}\sum_{i=1}^{n}D_{i}(Y_{i}-\hat{\beta}%
			_{0}(X_{i}))}{\frac{1}{n}\sum_{i=1}^{n}D_{i}}-\frac{E^{\ast}[\pi^{\ast
			}(X)\beta(X)]}{\pi^{\ast}}\right) \\
		=  &  \sqrt{n}\left(  \frac{\frac{1}{n}\sum_{i=1}^{n}\left(  D_{i}\beta
			(X_{i})-E^{\ast}[\pi^{\ast}(X)\beta(X)]+D_{i}\varepsilon_{1i}-\frac{\pi^{\ast
				}(X_{i})}{1-\pi^{\ast}(X_{i})}(1-D_{i})\varepsilon_{0i}\right)  }{\pi^{\ast}%
		}\right) \\
		&  -\frac{E^{\ast}[\pi^{\ast}(X)\beta(X)]}{(\pi^{\ast})^{2}}\frac{1}{\sqrt{n}%
		}\sum_{i=1}^{n}(D_{i}-\pi^{\ast})+o_{p}(1)\\
		=  &  \frac{1}{\sqrt{n}}\sum_{i=1}^{n}\left(  \frac{D_{i}\beta(X_{i})}%
		{\pi^{\ast}}-\gamma-\frac{\gamma}{\pi^{\ast}}(D_{i}-\pi^{\ast})+\frac
		{D_{i}\varepsilon_{1i}-\frac{\pi^{\ast}(X_{i})}{1-\pi^{\ast}(X_{i})}%
			(1-D_{i})\varepsilon_{0i}}{\pi^{\ast}}\right)  +o_{p}(1)\\
		=  &  \frac{1}{\sqrt{n}}\sum_{i=1}^{n}\left(  \frac{D_{i}(\beta(X_{i}%
			)-\gamma)}{\pi^{\ast}}+\frac{D_{i}\varepsilon_{1i}-\frac{\pi^{\ast}(X_{i}%
				)}{1-\pi^{\ast}(X_{i})}(1-D_{i})\varepsilon_{0i}}{\pi^{\ast}}\right)
		+o_{p}(1),
	\end{align*}
	where $\gamma=E[\pi(X)\beta(X)]/\pi=E^{\ast}[\pi^{\ast}(X)\beta(X)]/\pi^{\ast
	}$ is the ATT. This implies that the asymptotic variance of $\sqrt{n}%
	(\hat{\gamma}-\gamma)$ is equal to (\ref{ATT-var}).
\end{proof}

\begin{proof}
	[Proof of Proposition \ref{prop:ATT-PS}]We now consider the ATT estimator in
	(\ref{ATT-PS-estimator}). Lemma \ref{lem:ATT-PS-adjust} shows that the
	adjustment for the estimation of $\pi^{\ast}(X)$ in the numerator of
	(\ref{ATT-PS-estimator}) is equal to
	\begin{equation}
		-\frac{1}{1-\pi^{\ast}(X)}(\beta_{0}(X)-\delta_{0}^{\ast}(\pi^{\ast
		}(X)))(D-\pi^{\ast}(X)). \label{ATT-PS-pi}%
	\end{equation}
	Using a similar line of argument leading to (\ref{ATT-num-influence}), we can
	show that the infeasible estimator $\frac{1}{n}\sum_{i=1}^{n}D_{i}(Y_{i}%
	-\hat{\delta}_{0}(\pi^{\ast}(X_{i})))$ with known $\pi^{\ast}(X)$ has the
	influence function
	\begin{align}
		&  D(Y-\delta_{0}^{\ast}(\pi^{\ast}(X)))-E^{\ast}[\pi^{\ast}(X)\delta^{\ast
		}(\pi^{\ast}(X))]-\frac{\pi^{\ast}(X)}{1-\pi^{\ast}(X)}(1-D)(Y-\delta
		_{0}^{\ast}(\pi^{\ast}(X)))\nonumber\\
		=  &  D\delta^{\ast}(\pi^{\ast}(X))-E^{\ast}[\pi^{\ast}(X)\delta^{\ast}%
		(\pi^{\ast}(X))]+D\tilde{\varepsilon}_{1}-\frac{\pi^{\ast}(X)}{1-\pi^{\ast
			}(X)}(1-D)\tilde{\varepsilon}_{0}, \label{ATT-PS1}%
	\end{align}
	where $\tilde{\varepsilon}_{d}=Y-\delta_{d}^{\ast}(\pi^{\ast}(X))$. Rewrite
	the adjustment in (\ref{ATT-PS-pi}) as
	\begin{align}
		&  -\left(  D-\frac{\pi^{\ast}(X)(1-D)}{1-\pi^{\ast}(X)}\right)  (\beta
		_{0}(X)-\delta_{0}^{\ast}(\pi^{\ast}(X)))\nonumber\\
		=  &  D(\beta(X)-\delta^{\ast}(\pi^{\ast}(X)))-D(\beta_{1}(X)-\delta_{1}%
		^{\ast}(\pi^{\ast}(X)))\nonumber\\
		&  +\frac{\pi^{\ast}(X)}{1-\pi^{\ast}(X)}(1-D)(\beta_{0}(X)-\delta_{0}^{\ast
		}(\pi^{\ast}(X))). \label{ATT-PS2}%
	\end{align}
	Because $\tilde{\varepsilon}_{d}=\varepsilon_{d}+\beta_{d}(X)-\delta_{d}%
	^{\ast}(\pi^{\ast}(X))$, summing up (\ref{ATT-PS1}) and (\ref{ATT-PS2}) we
	obtain the overall influence function of the numerator in
	(\ref{ATT-PS-estimator}) as
	\[
	D\beta(X)-E^{\ast}[\pi^{\ast}(X)\delta^{\ast}(\pi^{\ast}(X))]+D\varepsilon
	_{1}-\frac{\pi^{\ast}(X)}{1-\pi^{\ast}(X)}(1-D)\varepsilon_{0},
	\]
	which is identical to the influence function (\ref{ATT-num-influence}) of the
	numerator in (\ref{ATT-estimator}).
\end{proof}

\begin{proof}
	[Proof of Proposition \ref{prop:weighting}]Because the conditional
	distribution of $(Y,Z,D)$ is the same in the population and stratified sample,
	we have $E[\varphi(Y,Z,D)|D]=E^{\ast}[\varphi(Y,Z,D)|D]$. By iterated
	expectations
	\begin{align*}
		&  E^{\ast}\left[  \left(  D\frac{\pi}{\pi^{\ast}}+(1-D)\frac{1-\pi}%
		{1-\pi^{\ast}}\right)  \varphi(Y,Z,D)\right] \\
		=  &  E^{\ast}\left[  \left.  \left(  D\frac{\pi}{\pi^{\ast}}+(1-D)\frac
		{1-\pi}{1-\pi^{\ast}}\right)  \varphi(Y,Z,D)\right\vert D=1\right]  \pi^{\ast
		}\\
		&  +E^{\ast}\left[  \left.  \left(  D\frac{\pi}{\pi^{\ast}}+(1-D)\frac{1-\pi
		}{1-\pi^{\ast}}\right)  \varphi(Y,Z,D)\right\vert D=0\right]  (1-\pi^{\ast})\\
		=  &  E^{\ast}[\varphi(Y,Z,D)|D=1]\pi+E^{\ast}[\varphi(Y,Z,D)|D=0](1-\pi)\\
		=  &  E[\varphi(Y,Z,D)|D=1]\pi+E[\varphi(Y,Z,D)|D=0](1-\pi)\\
		=  &  E[\varphi(Y,Z,D)].
	\end{align*}
	
\end{proof}


\end{document}